\documentclass[a4paper,fleqn]{article}
\usepackage{a4wide}
\usepackage{amsmath,amssymb,amsthm}
\usepackage{thmtools}
\usepackage{enumerate}
\usepackage{graphicx,tikz}
\usetikzlibrary{arrows,calc}
\usepackage{abstract}
\usepackage{authblk}

\usepackage{hyperref}
\usepackage[T1]{fontenc}
\usepackage{mathpazo}
\makeatletter\@ifpackageloaded{mathpazo}\@tempswatrue\@tempswafalse
\if@tempswa
  \DeclareFontFamily{OT1}{pzc}{}
  \DeclareFontShape{OT1}{pzc}{m}{it}{<-> s * [1.15] pzcmi7t}{}
  \DeclareMathAlphabet{\mathpzc}{OT1}{pzc}{m}{it}
\fi\makeatother
\usepackage{bm}
\usepackage[utf8]{inputenc}

\usepackage[sorting=nyt,giveninits=true,maxbibnames=5,backend=biber]{biblatex}
\makeatletter\@ifpackageloaded{biblatex}{%
  \usepackage{csquotes} 
  \bibliography{references}
  \renewbibmacro{in:}{%
    \ifentrytype{incollection}{\printtext{\bibstring{in}\intitlepunct}}{}}
  \renewbibmacro{publisher+location+date}{%
    \iflistundef{publisher}
      {\setunit*{\addcomma\space}}
      {\setunit*{\addcomma\space}}%
    \printlist{publisher}%
    \setunit*{\addcomma\space}%
    \printlist{location}%
    \setunit*{\addcomma\space}%
    \usebibmacro{date}%
    \newunit}
  \DeclareFieldFormat[article]{pages}{#1\isdot}
  \DeclareFieldFormat[article,incollection,inproceedings,unpublished,eprint]{title}{#1\isdot}
  \DeclareFieldFormat[thesis]{title}{\mkbibemph{#1\isdot}}
  \DeclareFieldFormat[unpublished]{date}{(#1)\isdot}
  \DeclareFieldFormat[unpublished]{note}{#1\nopunct} 
  \DeclareFieldFormat[eprint]{date}{(#1)\isdot}
  \DeclareFieldFormat[eprint]{note}{#1\nopunct} 
  \DeclareFieldFormat[article]{journaltitle}{\mkbibemph{#1\isdot}}
  
  \AtEveryBibitem{%
    \ifentrytype{book}{}{
      \clearname{editor}
    }
  }
  \newbibmacro*{bbx:parunit}{%
    \ifbibliography
      {\setunit{\bibpagerefpunct}\newblock
       \usebibmacro{pageref}%
       \clearlist{pageref}%
       \setunit{\adddot\par\nobreak}}
      {}
  }
  \renewbibmacro*{doi+eprint+url}{%
    \usebibmacro{bbx:parunit}
    \iftoggle{bbx:doi}
      {\printfield{doi}}
      {}%
    \iftoggle{bbx:eprint}
      {\usebibmacro{eprint}}
      {}%
    \iftoggle{bbx:url}
      {\usebibmacro{url+urldate}}
      {}
  }
  \renewbibmacro*{eprint}{%
    \usebibmacro{bbx:parunit}
    \iffieldundef{eprinttype}
      {\printfield{eprint}}
      {\printfield[eprint:\strfield{eprinttype}]{eprint}}
  }
  \renewbibmacro*{url+urldate}{%
    \usebibmacro{bbx:parunit}
    \printfield{url}%
    \iffieldundef{urlyear}
      {}
      {\setunit*{\addspace}%
       \printtext[urldate]{\printurldate}}
  }
}{}\makeatother

\declaretheorem[numberwithin=section,refname={theorem,theorems},Refname={Theorem,Theorems}]{theorem}
\declaretheorem[sibling=theorem,style=definition]{definition}
\declaretheorem[sibling=theorem,name=Lemma]{lemma}
\declaretheorem[sibling=theorem,name=Proposition]{proposition}

\declaretheorem[sibling=theorem,style=definition,name=Example]{example}

\declaretheorem[numbered=no,name=Question]{question}

\makeatletter\@ifpackageloaded{hyperref}{%
  \usepackage{xcolor}
  \definecolor{dark-red}{rgb}{0.4,0.15,0.15}
  \definecolor{dark-blue}{rgb}{0.15,0.15,0.4}
  \definecolor{medium-blue}{rgb}{0,0,0.5}
  \hypersetup{
    colorlinks,
    linkcolor={dark-red},
    citecolor={dark-blue},
    urlcolor={medium-blue}%
  }

}{}\makeatother

\usepackage{sectsty}
\sectionfont{\large\bfseries}
\subsectionfont{\normalsize\bfseries}

\providecommand{\abs}[1]{\lvert#1\rvert}
\providecommand{\Abs}[1]{\left\lvert#1\right\rvert}
\providecommand{\norm}[1]{\lVert#1\rVert}

\providecommand{\floor}[1]{\lfloor#1\rfloor}
\providecommand{\Floor}[1]{\left\lfloor#1\right\rfloor}
\newcommand{\infw}[1]{%
  \ifcat\noexpand#1\relax\bm{#1}
  \else\mathbf{#1}\fi}          

\newcommand{\lags}[1]{\mathcal{L}_{#1}}

\newcommand{\abexp}[2]{\mathpzc{A\mkern-3mu e}_{#1}(#2)}
\newcommand{\act}[2]{\mathpzc{A\mkern-3mu c}_{#1}(#2)}

\newcommand{\Z}{\mathbb{Z}}
\newcommand{\R}{\mathbb{R}}
\newcommand{\T}{\mathbb{T}}

\newcommand{\keywords}[1]{\par\noindent{\footnotesize{\em Keywords\/}: #1}}

\begin{document}
  \title{On $k$-abelian Equivalence and Generalized Lagrange Spectra}
  \author[,1,2,3]{Jarkko Peltomäki\footnote{Corresponding author.\\E-mail addresses: \href{mailto:r@turambar.org}{r@turambar.org} (J. Peltomäki), \href{mailto:mawhit@utu.fi}{mawhit@utu.fi} (M. A. Whiteland).}}
  \affil[1]{The Turku Collegium for Science and Medicine TCSM, University of Turku, Turku, Finland}
  \affil[2]{Turku Centre for Computer Science TUCS, Turku, Finland}
  \affil[3]{University of Turku, Department of Mathematics and Statistics, Turku, Finland}
  \author[3]{Markus A. Whiteland}
  \date{}
  \maketitle
  \vspace{-1.5em}
  \noindent
  \hrulefill
  \begin{abstract}
    \vspace{-1em}
    \noindent
    We study the set of $k$-abelian critical exponents of all Sturmian words. It has been proven that in the case
    $k = 1$ this set coincides with the Lagrange spectrum. Thus the sets obtained when $k > 1$ can be viewed as
    generalized Lagrange spectra. We characterize these generalized spectra in terms of the usual Lagrange spectrum and
    prove that when $k > 1$ the spectrum is a dense non-closed set. This is in contrast with the case $k = 1$, where
    the spectrum is a closed set containing a discrete part and a half-line. We describe explicitly the least
    accumulation points of the generalized spectra. Our geometric approach allows the study of $k$-abelian powers in
    Sturmian words by means of continued fractions.

    \vspace{1em}
    \keywords{Sturmian word, $k$-abelian equivalence, Lagrange spectrum, continued fraction}
    \vspace{-1em}
  \end{abstract}
  \hrulefill

  \section{Introduction}
  The critical exponent of an infinite word $\infw{w}$ is the supremum of exponents of fractional powers occurring in
  $\infw{w}$. Famously Thue showed in 1906 \cite{1906:uber_unendliche_zeichenreihen} that the fixed point of the
  substitution $0 \mapsto 01$, $1 \mapsto 10$, now known as the Thue-Morse word for Morse's independent contribution
  \cite{1921:recurrent_geodesics_on_a_surface_of_negative_curvature}, has critical exponent $2$ meaning that it avoids
  powers with exponent at least $3$. The notion of critical exponent is central in the study of powers and their
  avoidance which have since Thue been a central theme in the area of combinatorics on words.

  Another important subject in combinatorics on words is the theory of Sturmian words. Sturmian words comprise a large
  class of extensively studied words with strong connections to number theory, particularly to continued fractions
  (see, e.g., \cite{2007:sturmian_and_episturmian_words}, \cite[Chapter~2]{2002:algebraic_combinatorics_on_words},
  \cite[Chapter~6]{2002:substitutions_in_dynamics_arithmetics_and_combinatorics} and the references therein). The
  powers occurring in Sturmian words are well-understood, and a formula for the critical exponent of a Sturmian word
  was determined by Damanik and Lenz \cite{2002:the_index_of_sturmian_sequences} and Justin and Pirillo
  \cite{2001:fractional_powers_in_sturmian_words}. For example, the critical exponent of the Fibonacci word, the fixed
  point of the substitution $0 \mapsto 01$, $1 \mapsto 0$, is $(5 + \sqrt{5})/2$ as was already derived in
  \cite{1992:repetitions_in_the_fibonacci_infinite_word}. The critical exponent of the Fibonacci word is minimal among
  all Sturmian words.

  In recent years, there has been a substantial amount of research in generalizations of the concept of a power. A
  popular generalization is that of an abelian power; other generalizations are $k$-abelian powers (see below) and
  those based on $k$-binomial equivalence \cite{2015:another_generalization_of_abelian_equivalence_binomial}. Two words
  $u$ and $v$ are abelian equivalent, written $u \sim_1 v$, if one is obtained from the other by permuting letters. If
  $u_0$, $u_1$, $\ldots$, $u_{n-1}$ are abelian equivalent words of length $m$, then their concatenation
  $u_0 u_1 \cdots u_{n-1}$ is an abelian power of exponent $n$ and period $m$ (only integer exponents are considered).
  Thus an abelian power is a generalization of the usual notion of a power: the abelian equality relation is used in place
  of the usual equality relation. Questions regarding abelian powers were already raised by Erd{\H o}s in 1957
  \cite{1957:some_unsolved_problems}. More recently there has been a burst of activity on the subject starting,
  perhaps, with the 2011 paper \cite{2011:abelian_complexity_of_minimal_subshifts} by Richomme, Saari, and Zamboni.
  See, e.g., the references of \cite{2016:abelian_powers_and_repetitions_in_sturmian_words} and especially the papers
  \cite{2013:abelian_returns_in_sturmian_words,2013:some_properties_of_abelian_return_words,2016:abelian_powers_and_repetitions_in_sturmian_words,2017:abelian-square-rich_words}
  related to Sturmian words.

  The first author studied with Fici et al. the abelian critical exponents of Sturmian words in
  \cite{2016:abelian_powers_and_repetitions_in_sturmian_words}, where it was shown that there are abelian powers of
  arbitrarily high exponent starting at each position of a Sturmian word, a result also obtained in
  \cite{2011:abelian_complexity_of_minimal_subshifts}. This means that directly generalizing the notion of a critical
  exponent to the abelian setting only in terms of the exponent does not produce a quantity of interest (at least for
  Sturmian words). Thus an alternative definition was adopted in
  \cite{2016:abelian_powers_and_repetitions_in_sturmian_words}. The abelian critical exponent of an infinite word
  $\infw{w}$ is defined as the quantity
  \begin{equation}\label{eq:informal}
    \limsup_{m \to \infty} \left\{ \frac{n}{m} : \text{$u$ is an abelian power of exponent $n$ and period $m$ occurring in $\infw{w}$} \right\}
  \end{equation}
  measuring the maximal ratio between the exponents and periods of abelian powers in $\infw{w}$. This alternative
  definition does lead to an interesting quantity. The abelian critical exponent of a Sturmian word can be finite or
  infinite, and again the Fibonacci word has minimal exponent; this time the value being $\sqrt{5}$.
  
  Surprisingly, the set of abelian critical exponents of all Sturmian words turns out to coincide with the Lagrange
  spectrum. The Lagrange constant of an irrational $\alpha$ is the infimum of the real numbers $\lambda$ such that for
  every $c > \lambda$ the inequality $\abs{\alpha - n/m} < 1/cm^2$ has only finitely many rational solutions $n/m$. The
  Lagrange constant $\lambda(\alpha)$ of $\alpha$ is computed as follows:
  \begin{align*}
    \lambda(\alpha) &= \limsup_{t \to \infty} (q_t \norm{q_t\alpha})^{-1} \\
                    &= \limsup_{t \to \infty} ([a_{t+1}; a_{t+2}, \ldots] + [0; a_t, a_{t-1}, \ldots, a_1]),
  \end{align*}
  where $[a_0; a_1, a_2, \ldots]$ is the continued fraction expansion of $\alpha$ and $(q_k)$ is the sequence of
  denominators of its convergents ($\norm{x}$ measures the distance of $x$ to the nearest integer). The connection here
  is that for a fixed Sturmian word, the number $n$ in \eqref{eq:informal}, when maximal, equals the integer part of
  $1/\norm{m\alpha}$ for a certain irrational $\alpha$ (for details, see \autoref{sec:preliminaries} and
  \autoref{lem:max_exponent}).
  
  The Lagrange spectrum is the set of finite Lagrange constants of irrational numbers. The Lagrange spectrum has been
  studied extensively, but many of its properties still remain a mystery. The spectrum has a curious structure: its
  initial part inside the interval $[\sqrt{5}, 3)$ is discrete as shown by Markov already in late 19th century
  \cite{1879:sur_les_formes_quadratiques_binaires_indefinies,1880:sur_les_formes_quadratiques_binaires_indefinies_ii},
  but it contains a half-line as was famously proven by Hall in 1947
  \cite{1947:on_the_sum_and_products_of_continued_fractions}. Good sources for information on the Lagrange spectrum are
  the monograph of Cusick and Flahive \cite{1989:the_markoff_and_lagrange_spectra} and Aigner's book
  \cite{2013:markovs_theorem_and_100_years_of_the_uniqueness}.

  Another relatively recent development in combinatorics on words is the systematic study of a generalization of
  abelian equivalence called $k$-abelian equivalence initiated by Karhumäki, Saarela, and Zamboni in
  \cite{2013:on_a_generalization_of_abelian_equivalence_and_complexity_of_infinite}. This generalization originally
  appears in a 1980 paper of Karhumäki \cite{1980:generalized_parikh_mappings_and_homomorphisms}. Two words $u$ and $v$
  are said to be $k$-abelian equivalent, written $u \sim_k v$, if $\abs{u}_w = \abs{v}_w$ for each nonempty word $w$ of
  length at most $k$ (here $\abs{u}_w$ stands for the number of occurrences of $w$ as a factor of $u$). Thus
  $1$-abelian equivalence is simply the abelian equivalence discussed above. The $k$-abelian equivalence relation is
  clearly an equivalence relation, but it is also a congruence relation. For $k = 1, 2, \ldots$, the corresponding
  $k$-abelian equivalence relations can be seen as refinements of the abelian equivalence relation approaching the
  usual equality relation. The $k$-abelian equivalence has been studied especially from the points of view of factor
  complexity and power avoidance; for more information, see the recent paper
  \cite{2017:on_growth_and_fluctuation_of_k_abelian_complexity} and its references.

  The purpose of the current paper is to generalize the research of
  \cite{2016:abelian_powers_and_repetitions_in_sturmian_words} on abelian critical exponents of Sturmian words to the
  $k$-abelian setting. That is, we use the general $k$-abelian equivalence in place of abelian equivalence to obtain
  the notion of $k$-abelian critical exponent and study the set $\lags{k}$ of $k$-abelian critical exponents of
  Sturmian words. As $\lags{1}$ is the Lagrange spectrum, the sets $\lags{k}$ for $k > 1$ can be seen as combinatorial
  generalizations of the Lagrange spectrum.

  Our main contribution is the characterization of the $k$-Lagrange spectrum $\lags{k}$ in terms of the Lagrange
  spectrum $\lags{1}$. Our result, \autoref{thm:main_relation}, states that the $k$-abelian critical exponent of a
  Sturmian word $\infw{s}$ with abelian critical exponent $K$ equals $cK$ for a particular constant $c$, $0 < c < 1$,
  which depends on $k$ and $\infw{s}$. The relation between $\lags{1}$ and $\lags{k}$ is thus quite simple. However,
  the sets $\lags{k}$ inherit the complicated structure of the Lagrange spectrum $\lags{1}$. We show that for $k > 1$
  we have $\lags{k} \subseteq (\sqrt{5}/(2k-1), \infty)$, the number $\sqrt{5}/(2k-1)$ being the least accumulation
  point of $\lags{k}$ (\autoref{thm:endpoints}). Moreover, we prove that the set $\lags{k}$ is dense in
  $(\sqrt{5}/(2k-1), \infty)$ (\autoref{thm:dense}). This contrasts the case $k = 1$ where the initial part of
  $\lags{1}$ is discrete. The set $\lags{1}$ is known to contain a half-line. We do not know if $\lags{k}$ contains an
  analogous half-line for $k > 1$; we leave this problem open.
  
  Our approach is to first give an arithmetical and geometric interpretation for what it means for two factors of a
  Sturmian word to be $k$-abelian equivalent and then to employ continued fractions to derive our results. This
  approach is similar to that of \cite{2016:abelian_powers_and_repetitions_in_sturmian_words} where the usage of
  continued fractions was crucial. The arithmetical interpretation complements the combinatorial methods of
  \cite{2013:on_a_generalization_of_abelian_equivalence_and_complexity_of_infinite}: we make some results of
  \cite{2013:on_a_generalization_of_abelian_equivalence_and_complexity_of_infinite} on Sturmian words more precise. Our
  approach also makes it possible to efficiently find the possible exponents and locations of $k$-abelian powers
  occurring in a given Sturmian word.

  The paper is organized as follows. In \autoref{sec:preliminaries}, we give the necessary definitions and background
  information on Sturmian words and number theory. After this we present the main results and their proofs in
  \autoref{sec:results}. \autoref{sec:examples} provides further discussion on some matters raised in
  \autoref{sec:results}. Finally, \autoref{sec:open_problems} concludes the paper with open problems.

  \section{Preliminaries}\label{sec:preliminaries}
  We use standard terminology from combinatorics on words; we refer the reader to
  \cite{2002:algebraic_combinatorics_on_words} for any undefined terms. The words considered in this paper are finite
  or infinite binary words over the alphabet $\{0, 1\}$. We distinguish infinite words from finite words by referring
  to them with boldface symbols. By $\abs{w}$ we mean the length of the finite word $w$. The $n^\text{th}$ \emph{power}
  of a finite word $w$ is the word obtained by repeating it consecutively $n$ times, and it is denoted by $w^n$. For
  the infinite repetition of $w$, we use the notation $w^\omega$. An infinite word is \emph{ultimately periodic} if it
  can be written in the form $uv^\omega$ for some finite words $u$ and $v$; otherwise it is \emph{aperiodic}.

  We denote by $\abs{w}_u$ the number of occurrences of the nonempty word $u$ as a factor of $w$. If $u$ and $v$ are
  finite words over an alphabet $A$, then $u$ and $v$ are \emph{abelian equivalent}, written $u \sim_1 v$, if
  $\abs{u}_a = \abs{v}_a$ for each letter $a$ of $A$. Let then $k$ be a fixed positive integer. We say that $u$ and $v$
  are \emph{$k$-abelian equivalent}, written $u \sim_k v$, if $\abs{u}_w = \abs{u}_w$ for each word $w$ of length at
  most $k$. Notice that if $k = 1$, then $k$-abelian equivalence is simply the abelian equivalence. For words of length
  at least $k - 1$ we can alternatively say that $u \sim_k v$ if and only if $u$ and $v$ have a common prefix and a
  common suffix of length $k - 1$ and $\abs{u}_w = \abs{u}_w$ for each word $w$ of length $k$
  \cite[Lemma~2.3]{2013:on_a_generalization_of_abelian_equivalence_and_complexity_of_infinite}. Thus, for words of
  length at most $2k-1$, the $k$-abelian equivalence is in fact the equality relation
  \cite[Lemma~2.4]{2013:on_a_generalization_of_abelian_equivalence_and_complexity_of_infinite}. The $k$-abelian
  equivalence relation is a congruence relation. If $u_0$, $u_1$, $\ldots$, $u_{n-1}$ are $k$-abelian equivalent words
  of length $m$, then their concatenation $u_0 u_1 \cdots u_{n-1}$ is a \emph{$k$-abelian power of exponent $n$ and
  period $m$}. In this paper, we consider only nondegenerate powers, that is, we assume that $n \geq 2$.

  Recall that every irrational real number $\alpha$ has a unique infinite continued fraction expansion:
  \begin{equation}\label{eq:cf}
    \alpha = [a_0; a_1, a_2, a_3, \ldots] = a_0 + \dfrac{1}{a_1 + \dfrac{1}{a_2 + \dfrac{1}{a_3 + \ldots}}}
  \end{equation}
  with $a_0 \in \Z$ and $a_t \in \Z_+$ for $t \geq 1$. The numbers $a_i$ are called the \emph{partial quotients} of
  $\alpha$. By cutting the expansion after $t + 1$ terms, we obtain a rational number
  $[a_0; a_1, a_2, a_3, \ldots, a_t]$, which we denote by $p_t / q_t$. These rationals $p_t / q_t$ are the
  \emph{convergents} of $\alpha$. The convergents of $\alpha$ satisfy the best approximation property, that is,
  \begin{equation*}
    \norm{q_t \alpha} = \min_{0 < m \leq q_{t+1}} \norm{m\alpha}
  \end{equation*}
  for all $t \geq 1$. Here $\norm{x}$ measures the distance of $x$ to the nearest integer. In other words,
  $\norm{x} = \min\{ \{x\}, 1 - \{x\} \}$, where $\{x\}$ denotes the fractional part of $x$. Two numbers with continued
  fraction expansions $[a_0; a_1, \ldots]$ and $[b_0; b_1, \ldots]$ are \emph{equivalent} if there exist integers $N$
  and $M$ such that $a_{N+i} = b_{M+i}$ for all $i \geq 0$. As we shall see later, continued fractions are useful in
  studying Sturmian words (defined below). More details on the connection to Sturmian words can be found, e.g., in
  \cite[Chapter~4]{diss:jarkko_peltomaki}.
  
  Let $\alpha$ be an irrational number, and define the \emph{Lagrange constant} $\lambda(\alpha)$ of $\alpha$ as the
  infimum of real numbers $\lambda$ such that for every $c > \lambda$ the inequality
  \begin{equation}\label{eq:ls}
    \Abs{\alpha - \frac{p}{q}} < \frac{1}{cq^2}
  \end{equation}
  has only finitely many rational solutions $p/q$. Famously Hurwitz's Theorem states that
  $\lambda(\alpha) \geq \sqrt{5}$ for any irrational $\alpha$, and there exists numbers with
  $\lambda(\alpha) = \sqrt{5}$. The numbers with finite Lagrange constant are often called \emph{badly approximable
  numbers} in the literature. The Lagrange constant of $\alpha$ with continued fraction expansion as in \eqref{eq:cf}
  is computed as follows:
  \begin{equation}\label{eq:lagrange}
    \lambda(\alpha) = \limsup_{t \to \infty} ([a_{t+1}; a_{t+2}, \ldots] + [0; a_t, a_{t-1}, \ldots, a_1]).
  \end{equation}
  From this formula, it is clear that two equivalent numbers have the same Lagrange constant. The \emph{Lagrange
  spectrum} is the set of finite Lagrange constants. This set has many curious properties, and we shall return to them
  later at the end of \autoref{ssec:k-lagrange_spectrum}. For details on the Lagrange spectrum, see
  \cite{1989:the_markoff_and_lagrange_spectra} or \cite{2013:markovs_theorem_and_100_years_of_the_uniqueness}.
  
  Sturmian words are defined as the codings of orbits of points in an irrational circle rotation with two intervals.
  This understanding is sufficient for our purposes, but many other viewpoints exist; see, e.g.,
  \cite{2002:substitutions_in_dynamics_arithmetics_and_combinatorics,2002:algebraic_combinatorics_on_words}. Identify
  the unit interval $[0,1)$ with the unit circle $\T$, and let $\alpha$ be a fixed irrational. The mapping
  $R\colon \T \to \T, \, x \mapsto \{x + \alpha\}$ defines a rotation on $\T$. Partition the circle $\T$ into two
  intervals $I_0$ and $I_1$ defined by the points $0$ and $\{1-\alpha\}$. Let $\nu$ be the coding function defined by
  setting $\nu(x) = 0$ if $x \in I_0$ and $\nu(x) = 1$ if $x \in I_1$. Define $\infw{s}_{x,\alpha}$ as the infinite
  word obtained by setting its $n^\text{th}, n \geq 0,$ letter to equal $\nu(R^n(x))$. The word $\infw{s}_{x,\alpha}$
  is called the \emph{Sturmian word of slope $\alpha$ and intercept $x$}.

  The above definition is not complete because we did not define how $\nu$ behaves in the endpoints $0$ and
  $\{1-\alpha\}$. There is some choice here, and we take either $I_0 = [0,\{1-\alpha\})$ and $I_1 = [\{1-\alpha\},1)$
  or $I_0 = (0,\{1-\alpha\}]$ and $I_1 = (\{1-\alpha\},1]$. These options are determined by whether or not $0 \in I_0$.
  This little detail makes no difference to us: only interior points of intervals are considered. Let $x, y \in \T$
  with $x < y$. Then by both $I(x, y)$ and $I(y, x)$ we mean the interval $[x, y)$ if $0 \in I_0$ and the interval
  $(x, y]$ if $0 \notin I_0$.

  One particular example of a Sturmian word is the Fibonacci word $\infw{f}$. Its slope is $1/\varphi^2$, where
  $\varphi$ is the golden ratio, and its intercept equals its slope. We have
  \begin{equation*}
    \infw{f} = 01001010010010100101001001010010 \cdots.
  \end{equation*}
  This word is also the fixed point of the substitution $0 \mapsto 01$, $1 \mapsto 0$.
  
  The sequence $(\{n\alpha\})_{n\geq 0}$ is dense in $[0,1)$ by Kronecker's Theorem, so Sturmian words of slope
  $\alpha$ have a common language $\mathcal{L}$ (the set of factors). Let $w$ denote a word $a_0 a_1 \cdots a_{n-1}$ of
  length $n$ in $\mathcal{L}$. Then there exists a unique subinterval $[w]$ of $\T$ such that the Sturmian word
  $\infw{s}_{x,\alpha}$ begins with $w$ if and only if $x \in [w]$. Clearly
  $[w] = I_{a_0} \cap R^{-1}(I_{a_1}) \cap \ldots \cap R^{-(n-1)}(I_{a_{n-1}})$ (here the choice of endpoints matters,
  but we only consider interior points of these intervals). The points $0$, $\{-\alpha\}$, $\{-2\alpha\}$, $\ldots$,
  $\{-n\alpha\}$ partition the circle into $n+1$ subintervals which are exactly the intervals $[w]$ for factors of
  length $n$. We call these $n+1$ intervals the \emph{level $n$ intervals}, and we denote the set containing them by
  $L(n)$. We abuse notation and write $\max L(n)$ (resp. $\min L(n)$) for the maximum (resp. minimum) length of a level
  $n$ interval.

  In the rest of this paper, we keep the slope $\alpha$ with continued fraction expansion $[a_0; a_1, a_2, \ldots]$
  fixed. Whenever we talk about the convergents $q_t$, the level $n$ intervals $L(n)$, the rotation $R$, etc., we
  implicitly understand that they relate to this fixed $\alpha$.
  
  \section{Main Results}\label{sec:results}

  \subsection{\texorpdfstring{$k$}{k}-abelian Equivalence in Sturmian Words}\label{ssec:eq_classes}
  Our first aim is to show that the $k$-abelian equivalence classes of factors of a Sturmian word correspond to certain
  intervals on the circle $\T$ and to characterize the endpoints of these intervals. We begin by recalling the
  following result of \cite{2013:on_a_generalization_of_abelian_equivalence_and_complexity_of_infinite} (specialized to
  Sturmian words).

  \begin{proposition}\label{prp:pref_suff_ab}
    \cite[Proposition~2.8]{2013:on_a_generalization_of_abelian_equivalence_and_complexity_of_infinite} Let $u$ and $v$
    be two factors of the same length occurring in some Sturmian word. Then $u \sim_k v$ if and only if they share a
    common prefix and a common suffix of length $\min\{\abs{u}, k-1\}$ and $u \sim_1 v$.
  \end{proposition}

  This result is interesting as it shows that rather weak conditions are enough for $k$-abelian equivalence in Sturmian
  words. This is not unique to Sturmian words: it holds for episturmian words
  \cite[Proposition~2.8]{2013:on_a_generalization_of_abelian_equivalence_and_complexity_of_infinite}, and in
  \cite[Theorem~1]{2018:on_the_k-abelian_complexity_of_the_cantor_sequence}, it is shown that
  \autoref{prp:pref_suff_ab} holds also for factors of the Cantor word, the fixed point of the substitution
  $0 \mapsto 000$, $1 \mapsto 101$. We will return to this matter in \autoref{sec:examples}.

  Let us then recall the following result which gives an arithmetical characterization of abelian equivalence in
  Sturmian words.

  \begin{proposition}\label{prp:abelian_characterization}
    \cite[Proposition~3.3]{2016:abelian_powers_and_repetitions_in_sturmian_words},
    \cite[Theorem~19]{2013:some_properties_of_abelian_return_words} Let $u$ and $v$ be two factors of the same length
    occurring in a Sturmian word of slope $\alpha$. Then $u \sim_1 v$ if and only if
    $[u], [v] \subseteq I(0, \{-\abs{u}\alpha\})$ or $[u], [v] \subseteq I(\{-\abs{u}\alpha\}, 1)$.
  \end{proposition}

  In other words, the two possible abelian equivalence classes for factors of length $m$ correspond to two intervals on
  the circle marked by the points $0$ and $\{-m\alpha\}$. Next we generalize \autoref{prp:abelian_characterization} for
  $k$-abelian equivalence.

  By \autoref{prp:pref_suff_ab}, we need to at least consider the prefixes and suffixes of length up to $k - 1$. Let
  $m \geq 1$, and define $\mathcal{D}_{k,m} = \{0, \{-\alpha\}, \{-2\alpha\}, \ldots, \{-\min\{m, k-1\}\alpha\}\}$.
  These points divide the circle into $\min\{m + 1, k\}$ intervals (which are the level $\min\{m, k-1\}$ intervals),
  and if points $x$ and $y$ belong to the same interval, then the prefixes of $\infw{s}_{x,\alpha}$ and
  $\infw{s}_{y,\alpha}$ of length $\min\{m, k-1\}$ are equal. Now if $m \geq k - 1$, then
  \begin{equation*}
    R^{-(m - (k - 1))}(\mathcal{D}_{k,m}) = \{\{-(m - (k - 1))\alpha\}, \ldots, \{-m\alpha\}\},
  \end{equation*}
  and these points also divide the circle into $k$ intervals. If $x$ and $y$ belong to the same interval, then the
  prefixes of $\infw{s}_{x,\alpha}$ and $\infw{s}_{y,\alpha}$ of length $m$ have a common suffix of length $k - 1$. Set
  $\mathcal{P}_{k,m} = \smash[t]{\mathcal{D}_{k,m} \cup R^{-(m - (k - 1))}(\mathcal{D}_{k,m})}$ if $m \geq k - 1$;
  otherwise set $\mathcal{P}_{k,m} = \mathcal{D}_{k,m}$.

  \begin{definition}
    $\mathcal{I}_{k,m}$ is the set of subintervals of $\T$ determined by the points of $\mathcal{P}_{k,m}$.
  \end{definition}
  
  What me mean by this precisely is that, to define the intervals $I_i$ of $\mathcal{I}_{k,m}$, we order the points
  $x_i$ of $\mathcal{P}_{k,m}$: $0 = x_0 < x_1 < \ldots < x_{\ell - 1} < x_\ell = 1$, $\ell = \abs{\mathcal{P}_{k,m}}$,
  and set $I_i = [x_i, x_{i+1})$ if $0 \in I_0$ and $I_i = (x_i, x_{i+1}]$ if $0 \notin I_0$ for $0 \leq i < \ell$.
  Observe that for $m < k-1$, the intervals $\mathcal I_{k,m}$ coincide with the level $m$ intervals.
  
  As before for the level $m$ intervals $L(m)$, by writing $\max \mathcal{I}_{k,m}$ we mean the maximum length of an
  interval in $\mathcal{I}_{k,m}$. We claim that the intervals $\mathcal{I}_{k,m}$ determine the $k$-abelian
  equivalence classes.

  \begin{theorem}\label{thm:eq_class_intervals}
    Let $u$ and $v$ be two factors of length $m$ occurring in a Sturmian word of slope $\alpha$. Then $u \sim_k v$ if
    and only if there exists $J \in \mathcal{I}_{k,m}$ such that $[u], [v] \subseteq J$.
  \end{theorem}
  \begin{proof}
    Assume that $m < k - 1$. Then $u \sim_k v$ if and only if $u = v$. This means that $[u]$ and $[v]$ equal one of the
    level $m$ intervals. When $m < k - 1$, the intervals $\mathcal{I}_{k,m}$ are precisely the level $m$ intervals, so
    we are done. We may thus assume that $m \geq k - 1$.

    Suppose first that there exists $J \in \mathcal{I}_{k,m}$ such that $[u], [v] \subseteq J$. By the definition of
    the intervals $\mathcal{I}_{k,m}$, the words $u$ and $v$ share a common prefix and a common suffix of length
    $k - 1$. Moreover they are abelian equivalent by \autoref{prp:abelian_characterization} because the point
    $\{-m\alpha\}$ separating the two abelian equivalence classes is among the points $\mathcal{P}_{k,m}$. Therefore
    \autoref{prp:pref_suff_ab} implies that $u \sim_k v$.

    Suppose that $u \sim_k v$. Then $u$ and $v$ share a common prefix and a common suffix of length $k - 1$. Assume for
    a contradiction that $[u]$ and $[v]$ are contained in distinct intervals of $\mathcal{I}_{k,m}$. Without loss of
    generality, we assume that $\sup[u] \leq \inf[v]$. Let $K$ be the interval containing exactly the points $z$ for
    which $\sup[u] \leq z \leq \inf[v]$. (If $\sup[u] = \inf[v]$, then we let $K$ to be the set containing the common
    endpoint of $[u]$ and $[v]$.) Since $[u]$ and $[v]$ are contained in distinct intervals of $\mathcal{I}_{k,m}$,
    there exists a point $x$ in $\mathcal{P}_{k,m}$ such that $x \in K$. Denote the set
    $\smash[t]{R^{-(m - (k - 1))}(\mathcal{D}_{k,m})}$ by $\mathcal{S}$. The point $x$ cannot be in $\mathcal{D}_{k,m}$
    because $u$ and $v$ share a common prefix of length $k - 1$. Therefore we must have $x \in \mathcal{S}$. Let $y$ be
    an arbitrary point in $\mathcal{S}$. If $y \in \mathbb{T} \setminus ([u] \cup [v] \cup K)$, then either
    $[u] \subseteq I(x,y)$ and $[v] \cap I(x, y) = \emptyset$ or symmetrically $[v] \subseteq I(x,y)$ and
    $[u] \cap I(x, y) = \emptyset$. Then, by the definition of the points $\mathcal{S}$, we see that $u$ and $v$ have
    distinct suffixes of length $k - 1$, which is impossible. We conclude that $\mathcal{S} \subseteq K$ (see
    \autoref{ex:fibonacci} for this situation). Since $\{-m\alpha\} \in \mathcal{S}$, it follows by
    \autoref{prp:abelian_characterization} that $u$ and $v$ are not abelian equivalent. This is a contradiction.
  \end{proof}

  Notice that $\mathcal{I}_{k,m}$ contains $2k$ intervals when $m \geq 2k - 1$ and $m + 1$ intervals when
  $0 \leq m \leq 2k - 2$. This number of abelian equivalence classes for factors of length $m$ characterizes Sturmian
  words; see \cite[Theorem~4.1]{2013:on_a_generalization_of_abelian_equivalence_and_complexity_of_infinite}.

  \begin{example}\label{ex:fibonacci}
    Let us consider the $2$-abelian equivalence classes of length $5$ of the Fibonacci word; its slope $\alpha$ is
    $1/\varphi^2$. On the left in \autoref{fig:Fibonacci_k-intervals_and_factors}, there are two concentric circles.
    The outer circle represents the level $5$ intervals separated by the points $0$, $\{-\alpha\} (\approx 0.62)$,
    $\{-2\alpha\}$ ($\approx 0.24$), $\{-3\alpha\}$ ($\approx 0.85$), $\{-4\alpha\}$ ($\approx 0.47$), and
    $\{-5\alpha\}$ ($\approx 0.09$). The inner circle shows the endpoints of the $2$-abelian equivalence classes. The
    points $0$ and $\{-\alpha\}$ of $\mathcal{D}_{2,5}$ are shown in black while the points $\{-4\alpha\}$ and
    $\{-5\alpha\}$ of $R^{-4}(\mathcal{D}_{2,5})$ are represented by circles filled with white. The concentric circles
    on the right in \autoref{fig:Fibonacci_k-intervals_and_factors} give the corresponding intervals and points when
    $m = 7$.

    We have $4$ $2$-abelian equivalence classes for length $5$: $\{00100\}$, $\{00101, 01001\}$, $\{01010\}$, and
    $\{10010, 10100\}$. The singleton classes are special. At the end of the proof of \autoref{thm:eq_class_intervals},
    we had to take some extra steps because factors corresponding to two distinct intervals of $\mathcal{I}_{k,m}$
    could share prefixes and suffixes of length $k - 1$. Indeed here $00100$ and $01010$ have common prefixes and
    suffixes of length $1$, but this does not guarantee abelian equivalence.
	\end{example}

  \begin{figure}
    \begin{minipage}{.5\textwidth}
      \centering
      \begin{tikzpicture}
        \newcommand\CircleRadiusI{1.8}
        \newcommand\CircleRadiusII{1.0}
        \newcommand\Slope{0.381966}

        \draw (0,0) circle (\CircleRadiusI);
        \foreach \n in {0,...,5}{
          \pgfmathparse{(-1)*\n*\Slope}
          \filldraw[fill=black] ({(\pgfmathresult - floor(\pgfmathresult)) * 360}:\CircleRadiusI) circle (1pt);
        }

        \draw (0,0) circle (\CircleRadiusII);
        \foreach \n/\c in {0/black, 1/black, 4/white, 5/white}{
          \pgfmathparse{(-1)*\n*\Slope}
          \filldraw[fill=\c] ({(\pgfmathresult - floor(\pgfmathresult)) * 360}:\CircleRadiusII) circle (1pt);
        }

        \node at (2.3,0.5)    {\scriptsize{$[00100]$}};
        \node at (1.5,1.7)    {\scriptsize{$[00101]$}};
        \node at (-1.7,1.5)   {\scriptsize{$[01001]$}};
        \node at (-2.25,-0.5) {\scriptsize{$[01010]$}};
        \node at (-0.2,-2.05) {\scriptsize{$[10010]$}};
        \node at (2.1,-1.0)   {\scriptsize{$[10100]$}};

        \node at (1.6,0) {\scriptsize{$0$}};
        \node at (-1.05,-1.1) {\scriptsize{$-\alpha$}};
        \node at (0.15,1.5) {\scriptsize{$-2\alpha$}};
        \node at (0.8,-1.2) {\scriptsize{$-3\alpha$}};
        \node at (-1.4,0.3) {\scriptsize{$-4\alpha$}};
        \node at (1.1,0.9) {\scriptsize{$-5\alpha$}};
      \end{tikzpicture}
    \end{minipage}%
    \begin{minipage}{.5\textwidth}
      \centering
      \begin{tikzpicture}
        \newcommand\CircleRadiusI{1.8}
        \newcommand\CircleRadiusII{1.0}
        \newcommand\Slope{0.381966}

        \draw (0,0) circle (\CircleRadiusI);
        \foreach \n in {0,...,7}{
          \pgfmathparse{(-1)*\n*\Slope}
          \filldraw[fill=black] ({(\pgfmathresult - floor(\pgfmathresult)) * 360}:\CircleRadiusI) circle (1pt);
        }

        \draw (0,0) circle (\CircleRadiusII);
        \foreach \n/\c in {0/black, 1/black, 6/white, 7/white}{
          \pgfmathparse{(-1)*\n*\Slope}
          \filldraw[fill=\c] ({(\pgfmathresult - floor(\pgfmathresult)) * 360}:\CircleRadiusII) circle (1pt);
        }

        \node at (2.4,0.5)   {\scriptsize{$[0010010]$}};
        \node at (1.5,1.8)   {\scriptsize{$[0010100]$}};
        \node at (-0.8,2.05) {\scriptsize{$[0100100]$}};
        \node at (-2.1,1.2)  {\scriptsize{$[0100101]$}};
        \node at (-2.3,-0.7) {\scriptsize{$[0101001]$}};
        \node at (-1.4,-1.8) {\scriptsize{$[1001001]$}};
        \node at (0.9,-2.0)  {\scriptsize{$[1001010]$}};
        \node at (2.2,-1.0)  {\scriptsize{$[1010010]$}};

        \node at (1.6,0) {\scriptsize{$0$}};
        \node at (-1.05,-1.1) {\scriptsize{$-\alpha$}};
        \node at (0.15,1.5) {\scriptsize{$-2\alpha$}};
        \node at (0.8,-1.2) {\scriptsize{$-3\alpha$}};
        \node at (-1.4,0.3) {\scriptsize{$-4\alpha$}};
        \node at (1.1,0.9) {\scriptsize{$-5\alpha$}};
        \node at (-0.4,-1.5) {\scriptsize{$-6\alpha$}};
        \node at (-0.7,1.35) {\scriptsize{$-7\alpha$}};
      \end{tikzpicture}
    \end{minipage}
    \caption{Factors of length $5$ and $7$ of the Fibonacci word on the unit circle. The outer circles illustrate the
             level $5$ and $7$ intervals and the inner circles the $2$-abelian equivalence classes of length $5$ and
             $7$.}
    \label{fig:Fibonacci_k-intervals_and_factors}
  \end{figure}
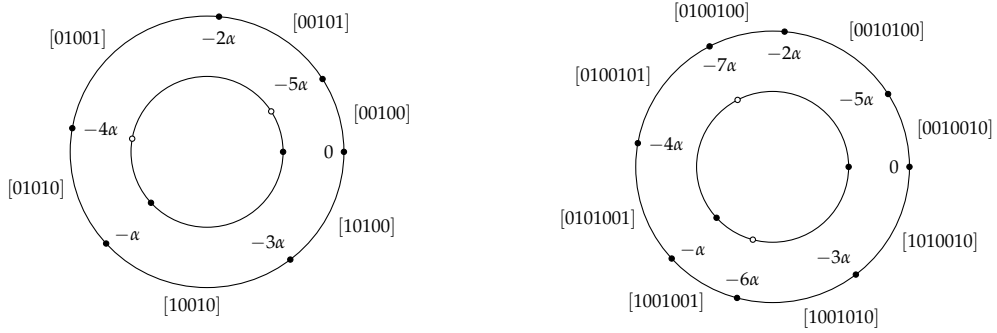

  We make an observation regarding the part of the proof of \autoref{thm:eq_class_intervals} showing that if two level
  $m$ intervals $[u]$ and $[v]$ are included in distinct intervals of $\mathcal{I}_{k,m}$ then $u \not\sim_k v$. The
  proof shows that if $u$ and $v$ have common prefixes and suffixes of length $k-1$, the only way that $u \not\sim_k v$
  is when all the points $\mathcal{S}$ (this is the set $R^{-(m - (k - 1))}(\mathcal{D}_{m,k})$) are contained in one
  level $(k - 1)$ interval $J$. We claim that this phenomenon cannot occur if $k \geq 2$ and
  $\norm{\alpha} > 1/(2(k-1))$. Notice that in the case $k = 2$ this may happen since $\norm{\alpha} < 1/2$ always.
  
  Assume that $k \geq 2$. There exist at least two points at distance $\norm{\alpha}$ in $\mathcal{S}$ (e.g.,
  $\{-(m-1)\alpha\}$ and $\{-m\alpha\}$) which implies that the length of $J$ is greater than $\norm{\alpha}$. If
  $k - 1 \geq \floor{1/\norm{\alpha}}$, then each interval determined by the points $0$, $\{-\alpha\}$, $\ldots$,
  $\{-(k-1)\alpha\}$ has length at most $\norm{\alpha}$, so we conclude that $k - 1 < \floor{1/\norm{\alpha}}$. The
  intervals determined by the points $0$, $\{-\alpha\}$, $\ldots$, $\{-(k-1)\alpha\}$ are now the same as those
  determined by the points $0$, $1 - \norm{\alpha}$, $1 - 2\norm{\alpha}$, $\ldots$, $1 - (k-1)\norm{\alpha}$, so all
  of them have length $\norm{\alpha}$ except one that has length $1 - (k-1)\norm{\alpha}$. Thus $J$ has length
  $1 - (k-1)\norm{\alpha}$. Since $R$ is an isometry, $J$ contains $(k - 1)$ intervals of length $\norm{\alpha}$
  (defined by the points of $\mathcal{S}$), and we must have $(k - 1)\norm{\alpha} < 1 - (k - 1)\norm{\alpha}$, that
  is, $\norm{\alpha} < 1/(2(k-1))$. Thus we obtain the following strengthening of \autoref{prp:pref_suff_ab}.

  \begin{theorem}\label{thm:pref_suff_ab_improved}
    Let $u$ and $v$ be two factors of the same length occurring in a Sturmian word of slope $\alpha$. Then $u \sim_k v$
    if and only if they share a common prefix and a common suffix of length $\min\{\abs{u}, k-1\}$ and $u \sim_1 v$.
    Moreover, the condition $u \sim_1 v$ may be omitted if $2(k - 1)\norm{\alpha} > 1$.
  \end{theorem}

  The slope of the Fibonacci word is approximately $0.38$, so \autoref{thm:pref_suff_ab_improved} says that the
  condition $u \sim_1 v$ can then be omitted when $k \geq 3$. It is rather surprising that such a weak condition is
  sufficient to establish $k$-abelian equivalence. This raises the question if it is possible to improve on the
  Fibonacci word and have an infinite word for which the condition is redundant even when $k = 2$. We study this
  question in \autoref{sec:examples}.
  
  \subsection{The \texorpdfstring{$k$}{k}-Lagrange Spectrum}\label{ssec:k-lagrange_spectrum}
  Let $\abexp{k,\alpha}{m}$ be the maximum exponent of $k$-abelian powers of period $m$ occurring in a Sturmian word of
  slope $\alpha$. We define the \emph{$k$-abelian critical exponent of slope $\alpha$} to be the quantity
  \begin{equation*}
    \limsup_{m\to\infty} \frac{\abexp{k,\alpha}{m}}{m},
  \end{equation*}
  and we denote it by $\act{k}{\alpha}$. It measures the maximal ratio between the exponent and period of a $k$-abelian
  power in a Sturmian word of slope $\alpha$; it was introduced in the case $k = 1$ in
  \cite{2016:abelian_powers_and_repetitions_in_sturmian_words} (in the current paper we follow the notation of the
  dissertation \cite{diss:jarkko_peltomaki} instead of the article
  \cite{2016:abelian_powers_and_repetitions_in_sturmian_words}). As mentioned in the introduction, the set of finite
  values of $\act{1}{\alpha}$ is the Lagrange spectrum
  \cite[Theorem~5.10]{2016:abelian_powers_and_repetitions_in_sturmian_words}, so the finite values of $\act{k}{\alpha}$
  can be viewed as a combinatorial generalization of the Lagrange spectrum. Thus we give the following definition.

  \begin{definition}
    The \emph{$k$-Lagrange spectrum} $\lags{k}$ is the set
    $\{\act{k}{\alpha} : \text{$\alpha$ is irrational}\} \cap \R$.
  \end{definition}

  In order to study $\lags{k}$, we begin by showing how to compute $\abexp{k,\alpha}{m}$ especially when $m$ is a
  denominator of a convergent of $\alpha$.

  Say that a Sturmian word $\infw{s}_{x,\alpha}$ of slope $\alpha$ and intercept $x$ begins with a $k$-abelian power of
  period $m$ and exponent $n$. The prefix of $\infw{s}_{x,\alpha}$ of length $m$ and the factor of
  $\infw{s}_{x,\alpha}$ of length $m$ starting after this prefix are $k$-abelian equivalent so, by
  \autoref{thm:eq_class_intervals}, the points $x$ and $\{x + m\alpha\}$ lie in a common interval of
  $\mathcal{I}_{k,m}$. The distance of these points is $\norm{m\alpha}$. Thus we see that the points $x$,
  $\{x + m\alpha\}$, $\ldots$, $\{x + (n - 1)m\alpha\}$ all lie in a common interval of $\mathcal{I}_{k,m}$, which must
  have length at least $(n - 1)\norm{m\alpha}$. Conversely, given such points, we see that the word
  $\infw{s}_{x,\alpha}$ begins with a $k$-abelian power of period $m$ and exponent $n$. Thus by considering the longest
  interval in $\mathcal{I}_{k,m}$, we obtain the following result (recall that $\max \mathcal{I}_{k,m}$ means the
  maximal length of an interval in $\mathcal{I}_{k,m}$).

  \begin{lemma}\label{lem:max_exponent}
    We have $\abexp{k,\alpha}{m} = \Floor{ \frac{\max \mathcal{I}_{k,m}}{\norm{m\alpha}} } + \gamma$, where $\gamma$ is
    $1$ if $\max \mathcal{I}_{k,m} \neq \norm{m\alpha}$ and $0$ otherwise.
  \end{lemma}

  \begin{example}
    (\autoref{ex:fibonacci} continued) The interval of the class $\{10010, 10100\}$ has length $\alpha$ which means by
    \autoref{lem:max_exponent} that using the words in the class a $2$-abelian power of period $5$ and exponent
    $\floor{\alpha/\norm{5\alpha}} + 1 = 5$ can be formed. Indeed, it is straightforward to check that
    $(10100)^2 (10010)^3$ is a factor of the Fibonacci word. Using words from the class $\{01010\}$ only $2$-abelian
    powers of exponent $\floor{\norm{3\alpha}/\norm{5\alpha}} + 1 = 2$ can be formed. The word $(00100)^2$ is not a
    factor of the Fibonacci word since it contains $000$. Indeed, we see using \autoref{lem:max_exponent} that the
    exponent for this class is $1$.

    Interestingly if $m = 7$, then the exponent for each equivalence class is $1$. The reason is that $\norm{7\alpha}$
    is large: we have $\norm{7\alpha} \approx 0.33$ whereas $\norm{5\alpha} \approx 0.09$. The $k$-abelian equivalence
    relation for $k > 1$ differs in this respect from abelian equivalence: it follows from
    \cite[Theorem~4.7]{2016:abelian_powers_and_repetitions_in_sturmian_words} that in any Sturmian word there exists an
    abelian square of period $m$ for each $m \geq 1$.
	\end{example}

  As the number $\max \mathcal{I}_{k,m}$ is generally difficult to find, let us argue next that when $m$ is chosen
  suitably then, in order to find $\abexp{k,\alpha}{m}$, it is sufficient to study the level $2k - 2$ intervals. As in
  \autoref{ssec:eq_classes}, the points
  $\mathcal{D}_{k,m} = \{0, \{-\alpha\}, \{-2\alpha\}, \ldots, \{-(k-1)\alpha\}\}$ together with the points
  $\mathcal{S} = R^{-(m - (k - 1))}(\mathcal{D}_{k,m}) = \{\{-(m - (k - 1))\alpha\}, \ldots, \{-m\alpha\}\}$ determine
  the intervals $\mathcal{I}_{k,m}$ of the $k$-abelian equivalence classes. Suppose now that $\norm{m\alpha}$ is
  sufficiently small. Then the points $R^m(\mathcal{S}) = R^{k-1}(D_{k,m})$ are close to the points $\mathcal S$. In
  fact, when comparing the intervals $\mathcal{I}_{k,m}$ defined by the points $\mathcal{D}_{k,m} \cup \mathcal{S}$ to
  those intervals defined by the points $\mathcal{D}_{k,m} \cup R^{k - 1}(\mathcal{D}_{k,m})$, we see that some
  intervals are shortened by $\norm{m\alpha}$ and some intervals are lengthened by $\norm{m\alpha}$, but the order of
  the points is the same whenever $\norm{m\alpha}$ is small enough. The points $\{-m\alpha\}$ and $0$ however merge,
  but this is irrelevant when considering $\max \mathcal{I}_{k,m}$ as we only lose a short interval of length
  $\norm{m\alpha}$. Now
  \begin{equation*}
    \mathcal{D}_{k,m} \cup R^{k - 1}(\mathcal{D}_{k,m}) = \{\{-(k-1)\alpha\}, \ldots, \{-\alpha\}, 0, \alpha, \ldots, \{(k-1)\alpha\}\}.
  \end{equation*}
  Using the fact that $R$ is an isometry, we can study the set
  $R^{-(k - 1)}(\mathcal{D}_{k,m} \cup R^{k - 1}(\mathcal{D}_{k,m}))$ instead. This set is the set of endpoints of the
  level $2k - 2$ intervals. It is quite obvious from the preceding that $\norm{m\alpha}$ is small enough whenever
  $\norm{m\alpha} < \min L(2k - 2)$. We have thus argued that whenever $\norm{m\alpha} < \min L(2k - 2)$, we have
  \begin{equation*}
    \abs{\max \mathcal{I}_{k,m} - \max L(2k - 2)} \leq \norm{m\alpha}.
  \end{equation*}
  Therefore we have proved the following lemma.

  \begin{lemma}\label{lem:approximate_exponent}
    Let $m$ be a positive integer and suppose that $\norm{m\alpha} < \min L(2k-2)$. Then
    \begin{equation*}
      \Abs{\Floor{\frac{\max L(2k-2)}{\norm{m\alpha}}} - \abexp{k,\alpha}{m}} \leq 1.
    \end{equation*}
  \end{lemma}

  This lemma shows that the exponents of $k$-abelian powers grow arbitrarily large (as we can make $\norm{m\alpha}$ as
  small as desired). A more general result was obtained in
  \cite[Theorem~5.4]{2013:on_a_generalization_of_abelian_equivalence_and_complexity_of_infinite}.

  With the results so far, we are able to show that for determining $\act{k}{\alpha}$ it is sufficient to consider
  $\abexp{k,\alpha}{m}$ only when $m$ is a denominator of a convergent. Recall that $q_t$ refers to the denominator of
  the $t^\text{th}$ convergent of $\alpha$.

  \begin{proposition}\label{prp:convergents_enough}
    For all large enough $t$, we have $\abexp{k,\alpha}{m} \leq \abexp{k,\alpha}{q_t} + 2$ for all
    $1 \leq m < q_{t+1}$.
  \end{proposition}
  \begin{proof}
    Let $t \geq 1$, and assume that $t$ be large enough so that $\norm{q_t\alpha} < \min L(2k-2)$. Suppose that $m$ is
    an integer such that $1 \leq m < q_{t+1}$. By the best approximation property of the convergents, we have
    $\norm{m\alpha} \geq \norm{q_t\alpha}$. Suppose first that $\norm{m\alpha} < \min L(2k-2)$. Then by
    \autoref{lem:approximate_exponent}, we have
    \begin{equation*}
      \abexp{k,\alpha}{m} \leq \frac{\max L(2k-2)}{\norm{m\alpha}} + 1 \leq \frac{\max L(2k-2)}{\norm{q_t\alpha}} + 1,
    \end{equation*}
    so, by the same lemma, we have $\abexp{k,\alpha}{m} \leq \abexp{k,\alpha}{q_t} + 2$. Suppose next that
    $\norm{m\alpha} \geq \min L(2k-2)$. Then
    \begin{equation*}
      \frac{\max L(m)}{\norm{m\alpha}} \leq \frac{\max L(m)}{\min L(2k-2)} \leq \frac{1}{\min L(2k-2)},
    \end{equation*}
    so $\abexp{k,\alpha}{m}$ is bounded by a constant. Thus $\abexp{k,\alpha}{m} < \abexp{k,\alpha}{q_t}$ for all large
    enough $t$. The sequence $(\abexp{k,\alpha}{q_i})_i$ reaches arbitrarily high values due to
    \autoref{lem:approximate_exponent}.
  \end{proof}

  \autoref{prp:convergents_enough} can be improved: $\abexp{k,\alpha}{m} \leq \abexp{k,\alpha}{q_t} + 1$ for all
  $1 \leq m < q_{t+1}$ and $t$ large enough. Proving this would complicate the argument significantly, and we do not
  need the improved statement in this paper. It is very well possible that
  $\abexp{k,\alpha}{q_t} > \abexp{k,\alpha}{q_{t+1}}$. For example, if $k = 2$ and say $\alpha = [0; 3, 1, 1, 1, 100,
  \overline{1}]$, then the sequence of denominators of convergents is $1$, $3$, $4$, $7$, $\ldots$, and it is readily
  computed that $\abexp{k,\alpha}{4} = 6 > 5 = \abexp{k,\alpha}{7}$. On the other hand, if $k = 1$, then we have
  $\abexp{k,\alpha}{m} < \abexp{k,\alpha}{q_t}$ for all $t$ and $1 \leq m < q_t$ as can be readily observed from
  \cite[Lemma~4.7]{2016:abelian_powers_and_repetitions_in_sturmian_words}.

  For $t$ large enough, let $m$ be an integer such that $q_t \leq m < q_{t+1}$. It follows from
  \autoref{prp:convergents_enough} that
  \begin{equation*}
    \frac{\abexp{k,\alpha}{m}}{m} \leq \frac{\abexp{k,\alpha}{q_t} + 2}{q_t},
  \end{equation*}
  so we can conclude using \autoref{lem:approximate_exponent} that
  \begin{equation*}
    \act{k}{\alpha} = \limsup_{t \to \infty} \frac{\abexp{k,\alpha}{q_t}}{q_t} = \limsup_{t \to \infty} \frac{\max L(2k-2)}{q_t\norm{q_t\alpha}}.
  \end{equation*}
  When $k = 1$, we obtain
  \begin{equation*}
    \act{1}{\alpha} = \limsup_{t \to \infty} \frac{1}{q_t\norm{q_t\alpha}},
  \end{equation*}
  so
  \begin{equation*}
    \act{k}{\alpha} = \max L(2k-2) \cdot \act{1}{\alpha}.
  \end{equation*}
  Let us restate the result.

  \begin{theorem}\label{thm:main_relation}
    We have $\act{k}{\alpha} = \max L(2k-2) \cdot \act{1}{\alpha}$ for all $k \geq 1$.
  \end{theorem}

  Notice that $\act{1}{\alpha}$ is finite if and only if $\alpha$ has bounded partial quotients; see
  \eqref{eq:lagrange}. Therefore $\act{k}{\alpha}$ is finite if and only if $\alpha$ has bounded partial quotients. As
  is well-known, numbers with bounded partial quotients comprise a set of measure zero.

  As mentioned in \autoref{sec:preliminaries}, equivalent numbers have the same Lagrange constant. By
  \autoref{thm:main_relation}, this is no longer true when $k > 1$ because $\max L(2k-2)$ depends on $\alpha$. It is
  not difficult to convince oneself that the points obtained in \autoref{thm:main_relation} from a single class of
  equivalent numbers form a dense set. This is what we shall prove next. As a corollary we obtain \autoref{thm:dense},
  which states that the $k$-Lagrange spectrum $\lags{k}$ is itself dense when $k > 1$. In the statement of the
  following lemma, by $\max L_\beta(\ell)$ we mean the maximal length of a level $\ell$ interval of slope $\beta$.

  \begin{lemma}\label{lem:density}
    Let $\alpha$ be irrational. The set $\{\max L_\beta(\ell) : \text{$\beta$ is equivalent to $\alpha$}\}$ is
    contained and dense in $(\tfrac{1}{\ell+1}, 1)$ for all $\ell > 1$.
  \end{lemma}
  \begin{proof}
    Clearly $\max L_\beta(\ell) > \tfrac{1}{\ell + 1}$ since there are $\ell + 1$ level $\ell$ intervals.
    Let $\gamma \in \smash[b]{(\tfrac{1}{\ell+1}, 1)}$, and suppose without loss of generality that it is irrational.
    By cutting the continued fraction expansion of $1-\gamma$ after finitely many partial quotients, we obtain a
    fraction that is as close as $1-\gamma$ as we desire. Thus we can find a rational $\beta$ such that $\ell \beta$ is
    arbitrarily close to $1-\gamma$ (from either side). Now form an irrational $\beta'$ by continuing the continued
    fraction expansion of $\beta$ in such a way that it is equivalent to $\alpha$. By selecting the partial quotients
    appropriately, we find that $\ell \beta'$ is arbitrarily close to $1-\gamma$. Consider now the level $\ell$
    intervals of slope $\beta'$. The longest such interval clearly has length $1-\ell\beta'$ since
    $\gamma > \tfrac{1}{\ell+1}$. As $1-\ell\beta'$ is as close to $\gamma$ as we like, the claim follows.
  \end{proof}

  As the smallest element of the Lagrange spectrum is $\sqrt{5}$, \autoref{thm:main_relation} and \autoref{lem:density}
  imply the following result.

  \begin{theorem}\label{thm:endpoints}
    Let $k > 1$. Then $\lags{k} \subseteq (\tfrac{\sqrt{5}}{2k-1}, \infty)$ and $\tfrac{\sqrt{5}}{2k-1}$ is the least
    accumulation point of $\lags{k}$. In particular, the set $\lags{k}$ is not closed.
  \end{theorem}
  
  This proposition should be compared with the fact that $\lags{1}$ is closed; cf.
  \cite[Theorem~2 of Chapter~3]{1989:the_markoff_and_lagrange_spectra}. Notice that it also follows that when $k > 1$,
  the Fibonacci word no longer has minimal critical $k$-abelian exponent among all Sturmian words.

  Let us then recall some remarkable facts about the Lagrange spectrum. \emph{Hall's ray} is the largest half-line
  contained in $\lags{1}$. It was proven by Hall that the half-line $[6, \infty)$ is contained in $\lags{1}$
  \cite{1947:on_the_sum_and_products_of_continued_fractions}. By series of improvements by several researchers, it was
  finally determined by Freiman \cite{1975:diophantine_approximation_and_geometry_of_numbers} that Hall's ray equals
  $[c_F, \infty)$, where $c_F$ is the \emph{Freiman constant}
  \begin{equation*}
    c_F = \frac{2221564096+283748\sqrt{462}}{491993569} = 4.5278295661 \ldots
  \end{equation*}
  The detailed history and references can be found in \cite[Chapter~4]{1989:the_markoff_and_lagrange_spectra}.
  Hall's result together with \autoref{thm:main_relation} and \autoref{lem:density} imply the following theorem.

  \begin{theorem}\label{thm:dense}
    The $k$-Lagrange spectrum $\lags{k}$ is dense in $(\tfrac{\sqrt{5}}{2k-1}, \infty)$ when $k > 1$.
  \end{theorem}
  \begin{proof}
    By \autoref{lem:density} and Hall's result, the intervals $(\tfrac{\sqrt{5}}{2k-1}, \sqrt{5})$ and
    $(\tfrac{c_F}{2k-1}, \infty)$ are dense with points of $\lags{k}$. Now $c_F$ is at most $6$, so
    $\smash[t]{\tfrac{c_F}{2k-1} \leq 2 < \sqrt{5}}$ meaning that these dense sets overlap.
  \end{proof}

  We do not know if $\lags{k}$ contains a half-line when $k > 1$. If true, it is not a straightforward consequence of
  Hall's and Freiman's results: the union of the dense subsets obtained from each $\theta \in [c_F, \infty)$ by
  \autoref{lem:density} is not automatically a half-line. This poses an interesting open problem.

  \begin{question}
    Does $\lags{k}$ contain a half-line when $k > 1$? If so, what is the largest such half-line? Is it
    $(\tfrac{c_F}{2k - 1}, \infty)$?
  \end{question}

  It is conceivable that a point in $\lags{1}$ below $c_F$ could map to $\tfrac{c_F}{2k-1}$. Moreover, $\lags{1}$ could
  contain an interval below $c_F$ (see below) that could produce an interval into $\lags{k}$.

  The usual Lagrange spectrum is not dense between $\sqrt{5}$ and $c_F$. In fact, substantial amount of research has
  been done on maximal gaps occurring in this interval, see for instance
  \cite[Chapter~5]{1989:the_markoff_and_lagrange_spectra}. It is known for example that the set
  $[\sqrt{5}, 3] \cap \lags{1}$ is discrete and that the interior of the interval $[\sqrt{12}, \sqrt{13}]$ does not
  include any points of $\lags{1}$ while its endpoints are in $\lags{1}$. It is unknown if $\lags{1}$ contains an
  interval below $c_F$. The existence of such an interval could show that $\lags{k}$ also contains an interval below
  $\tfrac{c_F}{2k-1}$, but it is plausible that this could also happen for other reasons. For example, it is possible
  for uncountably many numbers to have the same Lagrange constant, so an interval could be produced by means of
  \autoref{lem:density}. One such example is the number $3$; it is the Lagrange constant of uncountably many numbers
  \cite[Theorem~3, Chapter~IV§6]{1992:continued_fractions}. We do not believe that this particular example would
  provide an interval; we just mention it as a possibility. It is known that the part of $\lags{1}$ below
  $\sqrt{689}/8$ has measure zero \cite{1982:hausdorff_dimensions_of_cantor_sets}. It seems to us that studying
  intervals in $\lags{k}$ for $k > 1$ is of comparable difficulty as the study of intervals in $\lags{1}$.

  Let us also point out that it is easy to come up with numbers greater than $\sqrt{5}/(2k-1)$ that are not in
  $\lags{k}$. The two smallest elements of $\lags{1}$ are $\sqrt{5}$ and $\sqrt{8}$, so any point in $\lags{k}$ between
  $\sqrt{5}/(2k-1)$ and $\sqrt{8}/(2k-1)$ is of the form $\max L_\alpha(2k-2) \cdot \sqrt{5}$ for some $\alpha$
  equivalent to the golden ratio. The number $\max L_\alpha(2k-2)$ is always irrational, so rational multiples of
  $\sqrt{5}$ between $\sqrt{5}/(2k-1)$ and $\sqrt{8}/(2k-1)$ are not in $\lags{k}$.

  \subsection{The Spectrum \texorpdfstring{$\lags{\infty}$}{L}}
  As mentioned in the introduction, when the critical exponent is considered for the equality relation, it is typical
  to just measure the supremum of fractional exponents, not the ratio of the exponent and the period. In this final
  subsection, we briefly remark what happens if we look at the ratio instead.

  Analogous to what we have done already, we set
  \begin{equation*}
    \act{\infty}{\alpha} = \limsup_{m\to\infty} \frac{\abexp{\infty,\alpha}{m}}{m},
  \end{equation*}
  where $\abexp{\infty,\alpha}{m}$ is the maximum integer exponent of a power of period $m$ occurring in a Sturmian
  word of slope $\alpha$. We further set
  $\lags{\infty} = \{\act{\infty}{\alpha} : \text{$\alpha$ is irrational}\} \cap \R$. We show next that the set
  $\lags{\infty}$ contains every nonnegative real number.

  \begin{proposition}
    We have $\lags{\infty} = \R_{\geq 0}$.
  \end{proposition}
  \begin{proof}
    Consider powers occurring in a Sturmian word of slope $\alpha$ having continued fraction expansion
    $[0;a_1,a_2,\ldots]$ and sequence of convergents $(p_t/q_t)_t$. It is well-known that if $m$ is not a denominator
    of a convergent of $\alpha$, then any power of period $m$ has exponent at most $2$; see, e.g.,
    \cite[Lemma~3.6]{2002:the_index_of_sturmian_sequences} or \cite[Theorem~4.6.5]{diss:jarkko_peltomaki}. Moreover, if
    $m = q_t$ with $t > 1$, then the highest integer exponent of a power of period $m$ is $a_{t+1} + 2$
    \cite[Lemma~3.4]{2002:the_index_of_sturmian_sequences}, \cite[Theorem~4.6.5]{diss:jarkko_peltomaki}. Given that we
    have chosen the partial quotients $a_1$, $a_2$, $\ldots$, $a_t$ and thus determined the convergent $q_t$, we have
    complete freedom to choose $a_{t+1}$ to make the ratio $(a_{t+1} + 2)/q_t$ to behave the way we like.
    
    If the sequence $(a_t)_t$ of partial quotients is bounded, then we clearly have $\act{\infty}{\alpha} = 0$ because
    the sequence $(q_t)_t$ is increasing. Hence $0 \in \lags{\infty}$. Let then $\lambda$ be a fixed positive real
    number, and let $k_1$ be the least integer such that $k_1 > 1$ and that there exist nonnegative integers $r_1$ and
    $s_1$ such that $0 \leq s_1 < q_{k_1}$ and $\lambda - (r_1 + s_1/q_{k_1}) < \tfrac12$. Set $a_{1,1} = a_1$,
    $a_{1,2} = a_2$, $\ldots$, $a_{1,k_1} = a_{k_1}$, $a_{1,k_1 + 1} = \max\{1, q_{k_1}(r_1 + s_1/q_{k_1}) - 2\}$, and
    let $a_{1,t} = 1$ for $t > k_1 + 1$ to obtain a new number $\alpha_1$ with continued fraction expansion
    $[0; a_{1,1}, a_{1,2}, \ldots]$. Analogously, select then $k_2$ to be the least positive integer such that
    $k_2 > k_1$ and that there exist nonnegative integers $r_2$ and $s_2$ such that
    $\lambda - (r_2 + s_2/q_{1,k_2}) < \tfrac14$ where $q_{1,k_2}$ is the denominator of the $\smash[t]{k_2^\text{th}}$
    convergent of $\alpha_1$. Set $a_{2,1} = a_{1,1}$, $\ldots$, $a_{2,k_2} = a_{1,k_2}$,
    $a_{2,k_2 + 1} = \max\{1, q_{1,k_2}(r_2 + s_2/q_{1,k_2}) - 2\}$, and let $a_{2,t} = 1$ for $t > k_2 + 1$ to again
    obtain a number $\alpha_2$ with continued fraction expansion $[0; a_{2,1}, a_{2,2}, \ldots]$. Repeating this
    procedure yields sequences $(k_t)$, $(r_t)$, $(s_t)$ and a number $\beta$ with continued fraction expansion
    $[0; b_1, b_2, \ldots]$ and subsequence $(p'_t/q'_t)_t$ of its convergents such that
    \begin{equation*}
      \lambda - \frac{b_{k_t + 1} + 2}{q'_{k_t}} < \frac{1}{2^t}
    \end{equation*}
    for all $t \geq 1$ (the numbers $a_{t,k_t + 1}$ will grow arbitrarily large since $\lambda > 0$). We conclude that
    \begin{equation*}
      \limsup_{t \to \infty} \frac{\abexp{\infty,\beta}{q'_{k_t}}}{q'_{k_t}} = \lambda,
    \end{equation*}
    so $\act{\infty}{\beta} \geq \lambda$. As we have constructed the sequence $(b_t)_t$ in such a way that $b_t = 1$
    whenever $k_i < t < k_{i+1}$ for some $i$, it follows for such $i$ and $t$ large enough that
    \begin{equation*}
      \frac{b_t + 2}{q'_{t-1}} \leq \frac{b_{k_i} + 2}{q'_{t-1}} < \frac{b_{k_i} + 2}{q'_{k_i - 1}} \leq \lambda.
    \end{equation*}
    Therefore $\act{\infty}{\beta} = \lambda$ and $\lambda \in \lags{\infty}$.
  \end{proof}

  \section{Additional Questions}\label{sec:examples}
  At the end of \autoref{ssec:eq_classes}, we asked if there exists infinite words for which the condition of
  \autoref{thm:pref_suff_ab_improved} on abelian equivalence is redundant. The next proposition tells that such binary
  words exist but that they are rather uninteresting.

  \begin{proposition}\label{prp:no_binary}
    Let $\infw{w}$ be an infinite binary word such that for each of its factors $u$ and $v$ of equal length we have
    $u \sim_1 v$ if they share a common prefix and a common suffix of length $1$. Then $\infw{w}$ is ultimately
    periodic.
  \end{proposition}
  \begin{proof}
    Suppose for a contradiction that $\infw{w}$ is aperiodic, so either $00$ or $11$ occurs in $\infw{w}$. By symmetry,
    we assume that $00$ is a factor of $\infw{w}$. If $0011$ occurs also, then $001$ and $011$ occur. This is
    impossible as then by our assumption we should have $001 \sim_1 011$; this is clearly absurd. Thus $0010^n1$ occurs
    in $\infw{w}$ for some $n \geq 1$. The factors $000$ and $010$ are also incompatible, so $000$ cannot occur in
    $\infw{w}$. Hence $101$ and $1001$ are the only possible factors of the form $10^n 1$ with $n \geq 1$. Since
    $(100)^\omega$ is not a suffix of $\infw{w}$, either $101$ occurs or $10011$ must occur. The latter case we already
    ruled out, so $101$ occurs meaning that $111$ is not a factor of $\infw{w}$. If $11$ is not a factor, then
    $\infw{w}$ has a suffix that is a concatenation of the words $10$ and $100$. Suppose then that $11$ is a factor.
    The only way this is possible is that we have an occurrence of $1011$. This means that we do not see the
    incompatible factor $1001$. Hence $00$ occurs only as a prefix of $\infw{w}$. We have concluded that $\infw{w}$ has
    a suffix that is a product of the words $01$ and $011$. Thus by mapping $\infw{w}$ with the coding $0 \mapsto 1$,
    $1 \mapsto 0$, we obtain a word satisfying the assumptions and which has a suffix that is a product of $10$ and
    $100$. Thus without loss of generality, we may assume that $\infw{w}$ has a suffix that is a product of $10$ and
    $100$.

    If $100(10)^n100$ occurs in $\infw{w}$ for two distinct values of $n$, then for some $m\geq 0$ both $00(10)^m100$
    and $0(10)^{m+1}10$ are factors of $\infw{w}$. By our assumption, we must have $00(10)^m100 \sim_1 0(10)^{m+1}10$,
    but this is false. Therefore $100(10)^n100$ can occur only for a single value $n$, and $\infw{w}$ must have either
    of the words $(10)^\omega$ or $(100(10)^n)^\omega$ as a suffix. This is a contradiction. 
  \end{proof}

  However, if we allow more than two letters, then aperiodicity is possible as is shown by the next proposition. Let
  $A$ and $B$ be alphabets. Recall that a substitution $f\colon A^* \to B^*$ is a mapping such that $f(uv) = f(u)f(v)$.
  The image of the infinite word $a_0 a_1 \cdots$ under $f$ is the infinite word $f(a_0) f(a_1) \cdots$. If $w = uv$,
  then by $wv^{-1}$ we mean the word $u$. In the next proof, we need to know some properties of Sturmian words; these
  can be found in \cite[Chapter~2]{2002:algebraic_combinatorics_on_words}. Firstly, Sturmian words are \emph{balanced}.
  This means that for each two factors $u$ and $v$ of equal length occurring in some Sturmian word, we have
  $\abs{\abs{u}_0 - \abs{v}_0} \leq 1$. Secondly in a Sturmian word, there exists exactly one right special factor of
  length $n$ for all $n \geq 0$. A factor $u$ of an infinite word $\infw{w}$ is \emph{right special} if $ua$ and $ub$
  occur in $\infw{w}$ for distinct letters $a$ and $b$.

  Let $\sigma$ be the substitution defined by $\sigma(0) = 02$, $\sigma(1) = 1$. It is easy to see that the word
  $\sigma(\infw{s})$ is aperiodic for any Sturmian word $\infw{s}$.

  \begin{proposition}\label{prp:ternary_example}
    Let $k \geq 2$ and $\infw{s}$ be a Sturmian word containing $00$. Let $u$ and $v$ be two factors of the same length
    occurring in $\sigma(\infw{s})$. Then $u \sim_k v$ if and only if they share a common prefix and a common suffix of
    length $\min\{\abs{u}, k-1\}$.
  \end{proposition}
  \begin{proof}
    Suppose that $u$ and $v$ share a common prefix and a common suffix of length $\min\{\abs{u}, k-1\}$. We proceed as
    in the proof of \cite[Proposition~2.8]{2013:on_a_generalization_of_abelian_equivalence_and_complexity_of_infinite}
    (this is the proof of \autoref{prp:pref_suff_ab}). In this proof it is assumed that $u \sim_1 v$ and a counting
    argument is used to show that $u \sim_{\ell+1} v$ if $u \sim_\ell v$ for $1 \leq \ell < k$. By a careful analysis,
    it can be seen that this counting argument only uses the fact that there exists at most one right special factor of
    length $n$ for each $n$. Let $w$ and $w'$ be two right special factors of equal length occurring in
    $\sigma(\infw{s})$. It is clear that both $w$ and $w'$ must end with $2$. By the form of the substitution $\sigma$,
    there exist words $a$ and $b$ and unique factors $x$ and $y$ of $\infw{s}$ such that
    $a, b \in \{\varepsilon, 0\}$, $\abs{x} \geq \abs{y}$, $aw = \sigma(x)$, and $bw' = \sigma(y)$. Since $w$ and $w'$
    are right special, so are $x$ and $y$. It follows that $y$ is a suffix of $x$, so $w$ and $w'$ are suffixes of
    $\sigma(x)$. Since $\abs{w} = \abs{w'}$, they are equal. Thus we argued that $u \sim_k v$ if and only if they share
    a common prefix and a common suffix of length $\min\{\abs{u}, k-1\}$ and $u \sim_1 v$. Thus it suffices to show
    that $u \sim_1 v$.
    
    Like above, there exist words $a$ and $b$ and unique factors $x$ and $y$ of $\infw{s}$ such that
    $a \in \{\varepsilon, 0\}$, $b \in \{\varepsilon, 2\}$, $aub = \sigma(x)$, and $avb = \sigma(y)$. Let us show next
    that $x$ and $y$ are abelian equivalent. The claim follows from this. Since $k \geq 2$, the words $x$ and $y$ end
    in a common letter $c$. Now $x \sim_1 y$ if and only if $xc^{-1} \sim_1 yc^{-1}$ so, by replacing $x$ with
    $xc^{-1}$ and $y$ with $yc^{-1}$ if necessary, we may assume that $x$ and $y$ end with the letter $0$ ($1$ is
    always preceded by $0$ since $\infw{s}$ is balanced). For each binary word $w$, we have
    $\abs{\sigma(w)} = \abs{w} + \abs{w}_0$. Since $\abs{u} = \abs{v}$ (if $x$ and $y$ were replaced, we must replace
    $u$ and $v$ respectively by $\sigma(xc^{-1})$ and $\sigma(yc^{-1})$), we have
    \begin{equation}\label{eq:lengths}
      \abs{x} + \abs{x}_0 = \abs{y} + \abs{y}_0.
    \end{equation}
    Suppose without loss of generality that $\abs{x} \geq \abs{y}$, and write $x = zt$ with $\abs{z} = \abs{y}$. By
    plugging this into \eqref{eq:lengths}, we obtain that $\abs{t} + \abs{t}_0 = \abs{y}_0 - \abs{z}_0$. Since
    $\infw{s}$ is balanced, we see that $\abs{t} + \abs{t}_0 \leq 1$. Thus $t = \varepsilon$ or $t = 1$. The latter
    case is impossible as $x$ ends with $0$, so $t = \varepsilon$. Thus $\abs{x} = \abs{y}$ and so
    $\abs{x}_0 = \abs{y}_0$ by \eqref{eq:lengths}. This means that $x \sim_1 y$.
  \end{proof}

  Sturmian and episturmian words satisfy the property of \autoref{prp:pref_suff_ab} and it was shown in
  \cite{2018:on_the_k-abelian_complexity_of_the_cantor_sequence} that the Cantor word satisfies the property as well.
  The authors of \cite{2018:on_the_k-abelian_complexity_of_the_cantor_sequence} asked what sort of words satisfy this
  property. As we remarked above in the proof of \autoref{prp:ternary_example}, any infinite word having at most one
  right special factor of each length also satisfies this property. \autoref{prp:ternary_example} provides more
  examples of such words.

  \section{Further Open Problems}\label{sec:open_problems}
  It would be nice if our combinatorial generalization of the Lagrange spectrum had some number-theoretic
  interpretation, perhaps in connection to rational approximations of irrational numbers. We are unaware of such a
  connection.

  \begin{question}
    Is there an arithmetical characterization of the $k$-Lagrange spectrum $\lags{k}$?
  \end{question}

  An obvious open problem is to determine the $k$-abelian critical exponent of non-Sturmian infinite words. For
  example: what is the $k$-abelian critical exponent of the Tribonacci word, the fixed point of the substitution
  $0 \mapsto 012$, $1 \mapsto 02$, $2 \mapsto 0$? What about the Thue-Morse word? The case $k = 1$ is clear for the
  Thue-Morse word as the whole infinite word is an abelian power of infinite exponent and period $2$.

  Instead of looking at particular words or classes of words, it would be interesting to determine the set of critical
  exponents of all infinite words. In \cite{2007:every_real_number_greater_than_1}, Krieger and Shallit show that every
  real number greater than $1$ is a critical exponent of some infinite word. The result of Freiman shows that every
  real number greater than $c_F$ is the abelian critical exponent of some infinite word. Our result \autoref{thm:dense}
  shows that a dense subset of $(\tfrac{c_F}{2k-1}, \infty)$ is attainable as $k$-abelian critical exponents when
  $k > 1$. We are thus led to ask the following question.\footnote{The question has been solved in the positive; see
  \cite{2019:every_nonnegative_real_number_is_a_critical}.}

  \begin{question}
    Is every nonnegative real number the $k$-abelian critical exponent of some infinite word?
  \end{question}

  In \cite{2016:abelian_powers_and_repetitions_in_sturmian_words}, the abelian periods of factors of Sturmian words
  were studied (for definitions, see \cite{2016:abelian_powers_and_repetitions_in_sturmian_words}). It was proven for
  example that the abelian period of a factor of the Fibonacci word is always a Fibonacci number. Same sort of
  questions could be asked in the $k$-abelian setting for Sturmian words more generally. We have not attempted this
  study.

  \section*{Acknowledgments}
  We thank the referee for a careful reading of the paper, which improved the presentation.

  \printbibliography
	
\end{document}